\documentclass[11pt,a4paper]{article}
\usepackage{mathrsfs}
\usepackage{amsmath}
\usepackage{bbm}
\usepackage{amsmath,amsthm,amssymb,amscd}
\usepackage{latexsym}
\usepackage{hyperref}
\usepackage[numbers,sort&compress]{natbib}
\usepackage{hypernat}
\usepackage{geometry}
\usepackage{enumerate}
\usepackage[utf8]{inputenc}
\usepackage[english]{babel}
\usepackage{authblk}

\newcommand{\tr}{{\rm Tr\hskip -0.3em}~}

\DeclareMathOperator{\frechetdiff}{\mathit d}
\newcommand{\fd}[1]{\frechetdiff\hskip -0.3em{#1}}

\newtheorem{theorem}{Theorem}[section]

\newtheorem{corollary}[theorem]{Corollary}
\newtheorem{lemma}[theorem]{Lemma}
\newtheorem{proposition}[theorem]{Proposition}

\newtheorem{remark}[theorem]{Remark}

\theoremstyle{definition} \theoremstyle{remark}
\numberwithin{equation}{section}

\begin{document}

\title{Variational representations related to Tsallis relative entropy}
 \author[1]{Guanghua Shi}
 \author[2]{Frank Hansen}

  \affil[1]{{\footnotesize School of Mathematical Sciences, Yangzhou University, Yangzhou, Jiangsu, China. \qquad sghkanting@163.com}}

  \affil[2]{\footnotesize Institute for Excellence in Higher Education, Tohoku University, Sendai, Japan. frank.hansen@m.tohoku.ac.jp}

\maketitle

\begin{abstract} We develop variational representations for the deformed logarithmic and exponential functions and use them to obtain variational representations related to the quantum Tsallis relative entropy. We extend Golden-Thompson's trace inequality to deformed exponentials with deformation parameter $ q\in[0,1], $ thus complementing the second author's previous study of the cases with deformation parameter $ q\in[1,2]$ or $ q\in[2,3]. $
\end{abstract}
MSC2010: 94A17; 81P45; 47A63; 52A41.\\[1ex]
{\bf Keywords:} Gibbs variational priciple; Golden-Thompson inequality; Lieb's concavity theorem; Tsallis relative entropy; variational representations.

\section{Introduction}
The well-known concavity theorem by Lieb \cite[Theorem 6]{Lie73} states that the map
\begin{equation}\label{Lieb's concavity theorem}
A\rightarrow \tr \exp(L+\log A),
\end{equation}
for a fixed self-adjoint matrix $L,$
is concave in positive definite matrices. This theorem is the basis for the proof of strong subadditivity of the quantum mechanical entropy \cite{LieR73}, and it is also very important in random matrix theory \cite{Tro122}.

Lieb's concavity theorem is also closely related to the Golden-Thompson trace inequality. Recently, the second author \cite[Theorem 3.1]{Han14} proved concavity of the trace function
\[
A \rightarrow \tr \exp\bigl(H^*\log(A)H\bigr)
\]
in positive definite matrices, if
$ H $ is a contraction. This result led to multivariate generalisations of the Golden-Thompson trace inequality. Likewise, the second author \cite[Theorem 3.1]{Han15} studied convexity/concavity properties of the trace function
\[
(A_1,\ldots,A_k)\rightarrow \tr \exp_q\left(\sum_{i=1}^k H^*_i\log_q(A_i)H_i\right)
\]
for $ H_1^*H_1+\cdots+H_kH_k^*=1, $
where $ \exp_q $ denotes the deformed exponential function, respectively $ \log_q $ denotes the deformed logarithmic function, for the deformation parameter $ q\in [1,3]. $ This analysis led to a generalization of Golden-Thompson trace inequality for $q$-exponentials with $ q\in[1,3]. $

There is furthermore a close relationship between Lieb's concavity theorem (\ref{Lieb's concavity theorem}) and entropies.
In \cite{Tro121}, Tropp formulated a variational representation
\begin{eqnarray}
\tr \exp(L+\log A)= \max_{X> 0} \left\{\tr (L+I)X -D(X|A)\right\},
\end{eqnarray}
where $D(X|A)=\tr (X\log X-X\log A)$ denotes the quantum relative entropy.
This variational representation, together with convexity of the quantum relative entropy, enabled Tropp to give an elementary proof of Lieb's concavity theorem \cite[Theorem 6]{Lie73}. Tropp's variational representation can easily be inverted to obtain a variational representation
\begin{eqnarray}
D(X|A)= \max_{L}\left\{ \tr(L+I)X-\tr \exp(L+\log A)\right\}
\end{eqnarray}
of the quantum relative entropy.
The well-known Gibbs variational principle for the quantum entropy
states that
\begin{eqnarray} \log\tr\exp L=\max_{X>0, \tr X=1} \left\{\tr X L - \tr X \log X\right\},
\end{eqnarray}
and for $X>0$ and $\tr X=1,$
\begin{eqnarray} -S(X)=\max_{L} \left\{\tr X L - \log\tr \exp L\right\},
\end{eqnarray}
where $L$ is self-adjoint.
Other variational representations in terms of the quantum relative entropy were given by Hiai and Petz \cite[Lemma 1.2]{HP93}:
\begin{eqnarray}
\log \tr \exp(L+\log A)= \max_{X> 0, \tr X=1} \left\{ \tr LX-D(X|A)\right\},
\end{eqnarray}
and for $X>0$ and $ \tr X=1,$
\begin{eqnarray}
D(X|A)= \max_{L}\left\{ \tr LX-\log \tr \exp(L+\log A)\right\}.
\end{eqnarray}
See \cite{Yu} for the various relations between the above variational representations. Furuichi \cite{Fur06} extended the two representations above to the deformed logarithmic and exponential functions with parameter $q\in [1,2]$:
\begin{eqnarray}
\log_q \tr \exp_q(L+\log_q A)= \max_{X> 0,\, \tr X=1} \left\{ \tr LX^{2-q}-D_{2-q}(X|A)\right\},
\end{eqnarray}
and if $X>0 $ and $ \tr X=1,$
\begin{eqnarray}
D_{2-q}(X|A)= \max_{L}\left\{ \tr LX^{2-q}-\log_q \tr \exp_q(L+\log_q A)\right\},
\end{eqnarray}
where $L$ is self-adjoint and $D_{2-q}(X|A)$ denotes the Tsallis relative quantum entropy with parameter $q\in[1,2].$

In section 2, we consider variational representations related to the deformed exponential and logarithmic functions by making use of the tracial Young's inequalities. In section 3 we then derive variational representations related to the Tsallis relative entropies, which may be considered extensions of equation $(1.2).$ In section 4, we consider the generalization of the Gibbs variational representations and then tackle the variational representations related to the Tsallis relative entropy under the conditions  $X>0$ and $\tr X=1.$  Finally, in section 5, we extend Golden-Thompson's trace inequality to deformed exponentials with deformation parameter $ q\in[0,1].$

Throughout this paper, the deformed logarithm denoted $\log_q$ is defined by setting
\begin{eqnarray*}
\log_q x=\left\{
\begin{array}{ll}
\displaystyle\frac{x^{q-1}-1}{q-1} \quad &q\neq 1\\[2.5ex]
\log x  &q=1.
\end{array}\qquad\qquad x>0
\right.
\end{eqnarray*}
The deformed logarithm is also denoted the $q$-logarithm. The deformed exponential function or the $ q $-exponential is defined as the inverse function to the $ q$-logarithm. It is denoted by  $\exp_q$ and is given by the formula
\begin{eqnarray*}
\exp_q x=\left\{
\begin{array}{lll}
(x (q-1)+1)^{1/(q-1)},\qquad &x> -1/(q-1), \qquad &q>1\\[1.5ex]
(x (q-1)+1)^{1/(q-1)}, &x< -1/(q-1), \qquad &q<1\\[1.5ex]
\exp x, &x\in\mathbf R,\qquad &q=1.
\end{array}
\right.
\end{eqnarray*}

The Tsallis relative entropy $ D_{p}(X\mid Y) $ is for positive definite matrices $ X,Y  $ and $ p\in[0,1) $ defined, see  \cite{Tsa88}, by setting
\[
D_{p}(X\mid Y)= \frac{\tr (X-X^{p}Y^{1-p})}{1-p}=\tr X^{p}(\log_{2-p}X-\log_{2-p}Y).
\]
This expression converges for $ p\rightarrow 1 $ to the relative quantum entropy $D(X\mid Y)$
introduced by Umegaki \cite{Ume62}. It is known that the Tsallis relative entropy is non-negative for states \cite[Proposition 2.4]{FYK04}, see also \cite[Lemma 1]{HLS17} for a direct proof of the non-negativity.

\section{Variational representations for some trace functions}

We consider variational representations related to the deformed logarithm functions.

\begin{lemma}\label{2.1} For positive definite operators $X$ and $ Y$ we have
\[
\tr Y=\left\{
\begin{array}{ll}
\displaystyle\max_{X>0}\bigl\{\tr X-\tr X^{2-q}\left(\log_q X-\log_q Y\right)\bigr\}, \qquad &q\le 2\\[2.5ex]
\displaystyle\min_{X>0}\bigl\{\tr X-\tr X^{2-q}\left(\log_q X-\log_q Y\right)\bigr\},  &q>2.
\end{array}
\right.
\]
\end{lemma}

\begin{proof} For positive definite operators $X$ and $ Y,$  the tracial Young inequality states that
\begin{eqnarray*}
\tr X^{p}Y^{1-p} & \le & p \tr X+(1-p)\tr Y, \qquad p\in [0,1].
\end{eqnarray*}
As for the reverse tracial Young inequalities, we refer the readers to the proof of \cite[Lemma 12]{CFL18} from which we extracted the inequality 
\[
\tr X^s\le s \tr XY + (1-s)\tr Y^{-s/(1-s)}, \qquad 0<s<1.
\]
Replacing $s$ by $1/p$ and then replacing $X$ by $X^p$ and $Y$ by $Y^{1-p},$ it follows for $p>1$ that
\[
\tr X^pY^{1-p}\ge p \tr X + (1-p) \tr Y.
\]
It is also easy to see that the above inequality holds for $p<0.$ 
Thus it follows that
\begin{eqnarray*}
\tr Y\ge \tr X-\frac{\tr X-\tr X^pY^{1-p}}{1-p}\,, \qquad p\in[0,+\infty);
\end{eqnarray*}
and
\begin{eqnarray*}
\tr Y\le \tr X-\frac{\tr X-\tr X^pY^{1-p}}{1-p}\,, \qquad p\in(-\infty,0).
\end{eqnarray*}
For $X=Y$ the above inequalities become equalities, hence
\begin{eqnarray*}
\tr Y=\left\{
\begin{array}{ll}
\displaystyle\max_{X>0}\Bigl\{\tr X-\frac{\tr X-\tr X^pY^{1-p}}{1-p}\Bigr\}, \qquad &p\ge 0,\\[2ex]
\displaystyle\min_{X>0}\Bigl\{\tr X-\frac{\tr X-\tr X^pY^{1-p}}{1-p}\Bigr\}, \quad &p< 0.
\end{array}
\right.
\end{eqnarray*}
Setting $q=2-p,$ we obtain
\begin{eqnarray*}
\tr Y=\left\{
\begin{array}{ll}
\displaystyle\max_{X>0}\Bigl\{\tr X-\frac{\tr X^{2-q}\left(X^{q-1}-Y^{q-1}\right)}{q-1}\Bigr\}, \qquad &q\le 2,\\[2ex]
\displaystyle\min_{X>0}\Bigl\{\tr X-\frac{\tr X^{2-q}\left(X^{q-1}-Y^{q-1}\right)}{q-1}\Bigr\}, &q > 2.
\end{array}
\right.
\end{eqnarray*}
\end{proof}

\begin{theorem}\label{2.2} Let $H$ be a contraction.
For a positive definite operator $A$  we have the variational representations
\[
\begin{array}{l}
\tr \exp_q\left(H^*\log_q (A) H\right)\\[2.5ex]
=\left\{
\begin{array}{ll}
\displaystyle\max_{X>0}\bigl\{\tr X-\tr X^{2-q}\left(\log_q X-H^*\log_q (A) H\right)\bigr\}, \qquad &q\le 2,\\[2.5ex]
\displaystyle\min_{X>0}\bigl\{\tr X-\tr X^{2-q}\left(\log_q X-H^*\log_q (A) H\right)\bigr\}, &q>2.
\end{array}
\right.
\end{array}
\]
\end{theorem}

\begin{proof} Since $H$ is contraction, it follows for $ q>1 $ that
\[
H^*\log_q(A)H> \frac{-1}{q-1}\,.
\]
By setting
$Y= \exp_q\left(H^*\log_q (A) H\right)$ in Lemma 2.1, we obtain the conclusions in the case $ q\ge 1. $
For $ q<1 $
we have
\[
H^*\log_q(A)H< \frac{-1}{q-1}\,.
\]
Setting $Y= \exp_q\left(H^*\log_q (A) H\right)$ in Lemma 2.1, we obtain the conclusions for $ q< 1. $
\end{proof}

\begin{corollary}\label{2.3} Let $H$ be a contraction and consider the map
\begin{eqnarray*}
\varphi(A)= \tr \exp_q\left(H^*\log_q (A) H\right)
\end{eqnarray*}
defined in positive definite operators. The following assertions are valid:
\begin{enumerate}
\item[(i)]  $ \varphi(A) $ is concave for $ 0\le q<1, $    
\item[(ii)] $ \varphi(A) $ is concave for $ 1\le q\le 2, $ 
\item[(iii)] $ \varphi(A) $ is convex for  $ 2< q\le 3. $  
\end{enumerate}
\end{corollary}

\begin{proof} By calculation we obtain
\begin{eqnarray*}
&&\tr X-\tr X^{2-q}\left(\log_q X-H^*\log_q (A) H\right)
\\&=&
\left(1-\frac{1}{q-1}\right)\tr X+ \frac{1}{q-1}\left[\tr X^{2-q}(1-H^*H)+\tr X^{2-q}H^*A^{q-1}H\right].
\end{eqnarray*}
Under the assumption in $ (i), $ we have
\[
1<2-q\le 2,\quad -1\le q-1<0,\qquad (2-q)+(q-1)=1.
\]
By Ando's convexity theorem, the trace function $\tr X^{2-q}H^*A^{q-1}H$ is thus jointly convex in $(X,A).$ We also realize that $\tr X^{2-q}(1-H^*H)$ is convex in $X.$ Therefore,
\[
\tr X-\tr X^{2-q}\left(\log_q X-H^*\log_q (A) H\right)
\]
is jointly concave in $(X,A).$ Hence, by Theorem 2.2 and \cite[Lemma 2.3]{CLie08} we obtain that $ \varphi(A) $ is concave for $ 0\le q<1. $ Under the assumption in $ (ii), $ we have
\[
0\le 2-q\le 1,\qquad 0\le q-1\le 1,\qquad (2-q)+(q-1)=1.
\]
By Lieb's concavity theorem, the trace function $\tr X^{2-q}H^*A^{q-1}H$ is jointly concave in $(X,A).$ The expression
\[
\tr X-\tr X^{2-q}\left(\log_q X-H^*\log_q (A) H\right)
\]
is therefore also jointly concave in $(X,A).$ By Theorem 2.2 and \cite[Lemma 2.3]{CLie08} we obtain that $ \varphi(A) $ is concave for $ 1\le q\le 2. $ Under the assumption in $ (iii), $  we have
\[
-1\le 2-q< 0,\qquad 1< q-1\le 2,\qquad (q-1)+(2-q)=1.
\]
By Ando's convexity theorem, the trace function $\tr X^{2-q}H^*A^{q-1}H$ is jointly convex in $(X,A).$ Since obviously $\tr X^{2-q}(1-H^*H)$ is convex in $X,$ we obtain that
\[
\tr X-\tr X^{2-q}\left(\log_q X-H^*\log_q (A) H\right)
\]
is jointly convex in $(X,A).$ Hence  $ \varphi(A) $ is convex for $ 2< q\le 3, $
by Theorem 2.2 and \cite[Lemma 2.3]{CLie08}.
\end{proof}

\begin{remark}\label{2.4}
The second author \cite{Han15} proved
the cases $1\le q\le 2$ and $2\le q\le 3$ in the above corollary by another method.  The case $0\le q\le 1$ may be similarly proved by using that the trace function
\[
A\rightarrow \tr (H^* A^{q-1} H)^{1/(q-1)}
\]
is concave for $0\le q\le 1.$
\end{remark}

\begin{proposition}\label{2.5}
Let $ H $ be a contraction.
\begin{enumerate}
\item[(i)] If $1\le q\le 2,$ then for positive definite $A$ and self-adjoint $L$ such that
\begin{eqnarray*}
L+H^*\log_q (A) H>-\frac{1}{q-1}\,,
\end{eqnarray*}
we have the equality
\begin{eqnarray*}
&&\tr \exp_q (L+H^*\log_q (A) H)\\[1ex]&=&
\max_{X>0}\bigl\{\tr X+\tr X^{2-q}L - \tr X^{2-q} \left(\log_q X-H^*\log_q (A) H\right)\bigr\}.
\end{eqnarray*}

\item[(ii)] If $ q>2, $ then for positive definite $A$ and self-adjoint $L$ such that
\begin{eqnarray*}
L+H^*\log_q (A) H>-\frac{1}{q-1}\,,
\end{eqnarray*}
we have the equality
\begin{eqnarray*}
&&\tr \exp_q (L+H^*\log_q (A) H)\\[1ex]&=&
\min_{X>0}\bigl\{\tr X+\tr X^{2-q}L -\tr X^{2-q}\left(\log_q X-H^*\log_q (A) H\right)\bigr\}.
\end{eqnarray*}

\item[(iii)] If $ q< 1, $ then for positive definite $A$ and self-adjoint $L$ such that
\begin{eqnarray*}
L+H^*\log_q (A) H<-\frac{1}{q-1}\,,
\end{eqnarray*}
we have the equality
\begin{eqnarray*}
&&\tr \exp_q (L+H^*\log_q (A) H)\\[1ex]&=&\max_{X>0}\bigl\{\tr X+\tr X^{2-q}L-\tr X^{2-q}\left(\log_q X-H^*\log_q (A) H\right)\bigr\}.
\end{eqnarray*}
\end{enumerate}
\end{proposition}

\begin{proof} Under the assumptions of $(i),(ii)$ and $(iii),$
the expression
$
\exp_q(L+H^*\log_q(A)H)
$
is well-defined and positive definite.
By setting $ Y=\exp_q(L+H^*\log_q(A)H) $ in Lemma 2.1, we obtain $(i),(ii)$ and $(iii).$
\end{proof}

\begin{corollary}\label{2.6}
Let $H$ be a contraction, and let $ L $ be positive definite. The map
\[
A\rightarrow \tr \exp_q (L+H^*\log_q (A) H),
\]
defined in positive definite operators, is concave for $1\le q\le 2$ and convex for $ 2<q\le 3. $
The map
\[
A\rightarrow \tr \exp_q (-L+H^*\log_q (A) H),
\]
defined in positive definite operators, is concave for $ 0\le q< 1. $
\end{corollary}

\begin{proof}
If $1\le q\le 2,$ the map $ X\to \tr X^{2-q}L $ is concave.
By an argument similar to the proof of Corollary 2.3 (ii), we obtain that the expression
\[
\tr X-\tr X^{2-q}\left(\log_q X-H^*\log_q (A) H\right)
\]
is jointly concave in $ (X,A). $ Then obviously
\[
\tr X+\tr X^{2-q}L-\tr X^{2-q}\left(\log_q X-H^*\log_q (A) H\right)
\]
is jointly concave in $ (X,A). $
By Proposition 2.5 (i) and \cite[Lemma 2.3]{CLie08}
we then conclude that
\[
\tr \exp_q (L+H^*\log_q (A) H)
\]
is concave in $ A $ for $1\le q\le 2.$
The case for $q>2$ can be proved by a similar argument as above.
If $0\le q<1,$ then the map $ X\to \tr X^{2-q}L $ is convex.
By an argument similar to the proof of Corollary 2.3(i), we obtain that the expression
\[
\tr X-\tr X^{2-q}\left(\log_q X-H^*\log_q (A) H\right)
\]
is jointly concave in $ (X,A). $
Thus,
\[
\tr X-\tr X^{2-q}L-\tr X^{2-q}\left(\log_q X-H^*\log_q (A) H\right)
\]
is jointly concave in $ (X,A). $
By Proposition 2.5 (iii) and \cite[Lemma 2.3]{CLie08},
we then obtain that
\[
\tr \exp_q (-L+H^*\log_q (A) H)
\]
is concave in $ A $ for $0\le q<1.$
\end{proof}

Setting $q=1$ in Corollary 2.6 we obtain:

\begin{corollary}\label{2.7}
Let $ H $ be a contraction, and let $L$ be self-adjoint. The map
\[
A\rightarrow \tr \exp (L+H^*\log (A) H)
\]
is concave in positive definite operators.
\end{corollary}

\begin{proposition}\label{2.8} Let $H$ be a contraction, and let $ L $ be positive definite. The map
\[
A\rightarrow \tr \exp_q\left(L+H^*\log_r (A) H\right)
\]
is convex in positive definite operators for $q, r \in [2,3]$ with $r\ge q.$
\end{proposition}

\begin{proof}
Since  $H$ is a contraction and $r\ge q\ge 2$ we obtain the inequalities
\[
H^*\log_r (A)H\ge \frac{-H^* H}{r-1}\ge \frac{-1}{r-1}\ge \frac{-1}{q-1}\,.
\]
We may thus apply the deformed exponential and set $Y=\exp_q\left(L+H^*\log_r (A) H\right)$ in Lemma 2.1 to obtain
\begin{eqnarray*}
&&\tr \exp_q (L+H^*\log_r (A) H)\\[1ex]&=&
\min_{X>0}\bigl\{\tr X+\tr X^{2-q}L - \tr X^{2-q} \left(\log_q X-H^*\log_r (A) H\right)\bigr\}
\\&=&
\min_{X>0}\Bigl\{\Bigl(1-\frac{1}{q-1}\Bigr)\tr X+\tr X^{2-q}\Bigl(L+\frac{1}{q-1}-\frac{H^* H}{r-1}\Bigr)+\frac{1}{r-1}\tr X^{2-q}H^* A^{r-1} H \Bigr\},
\end{eqnarray*}
where by the assumptions
\[
2-q\in[-1,0],\qquad r-1\in [1,2],\qquad (r-1)+(2-q)\ge 1.
\]
By Ando's convexity theorem and  \cite[Lemma 2.3]{CLie08} we then get the desired conclusions.
\end{proof}

\section{Variational expressions related to Tsallis relative entropy}

\begin{theorem}\label{3.1}
Let $H$ be a contraction. For positive definite operators $X$ and $A$ the following assertions hold:
\begin{enumerate}

\item[(i)] For $1\le q\le 2$  we have the equality
\begin{eqnarray*}
&& \tr X^{2-q}\left(\log_q X-H^* \log_q(A)H\right)
\\[1ex]&=&
\max_{L>-H^*\log_q (A) H-(q-1)^{-1}}\left\{\tr X +\tr X^{2-q}L-\tr \exp_q (L+H^*\log_q (A) H)\right\}.
\end{eqnarray*}

\item[(ii)] For $q>2$  we have the equality
\begin{eqnarray*}
&& \tr X^{2-q}\left(\log_q X-H^* \log_q(A)H\right)
\\[1ex]&=&
\min_{L>-H^*\log_q (A) H-(q-1)^{-1}}\left\{\tr X +\tr X^{2-q}L-\tr \exp_q (L+H^*\log_q (A) H)\right\}.
\end{eqnarray*}

\item[(iii)] For $q<1$ we have the equality
\begin{eqnarray*}
&& \tr X^{2-q}\left(\log_q X-H^* \log_q(A)H\right)
\\[1ex]&=&
\max_{L<-H^*\log_q (A) H-(q-1)^{-1}}\left\{\tr X +\tr X^{2-q}L-\tr \exp_q (L+H^*\log_q (A) H)\right\}.
\end{eqnarray*}
\end{enumerate}
\end{theorem}

\begin{proof}
Under the assumptions in $(i)$ and the natural condition
\begin{eqnarray*}
L+H^*\log_q (A) H>-\frac{1}{q-1},
\end{eqnarray*}
ensuring that $\exp_q (L+H^*\log_q (A) H)$ makes sense,
 we set
\[
G(L)=\tr X +\tr X^{2-q}L-\tr \exp_q (L+H^*\log_q (A) H)
\]
and obtain that $G(L)$ is concave.
By Proposition 2.5 (i), we then obtain the inequality
\[
G(L)\le \tr X^{2-q}\left(\log_q X-H^* \log_q(A)H\right).
\]
Inserting $L_0=\log_q X-H^* \log_q (A) H $  yields
\[
G(L_0)=\tr X^{2-q}\left(\log_q X-H^* \log_q(A)H\right),
\]
such that $G(L)$ attains its maximum in $L_0.$
Thus we obtain
\[
\max_{L>-H^*\log_q (A) H-(q-1)^{-1}}G(L)= \tr X^{2-q}\left(\log_q X-H^* \log_q(A)H\right),
\]
which proves $ (i). $
The case $(ii)$  can be proved by a similar argument.
Under the assumptions in $(iii),$ and the condition
\[
L+H^*\log_q (A) H<-\frac{1}{q-1}\,,
\]
 we set
\[
G(L)=\tr X+\tr X^{2-q}L-\tr \exp_q (L+H^*\log_q (A) H)
\]
and obtain that $G(L)$ is concave.
By Proposition 2.5 (iii), we then obtain the inequality
\[
G(L)\le \tr X^{2-q}\left(\log_q X-H^* \log_q(A)H\right).
\]
Inserting $L_0=\log_q X-H^* \log_q (A) H $ yields
\[
G(L_0)=\tr X^{2-q}\left(\log_q X-H^* \log_q(A)H\right)
\]
such that $G(L)$ attains its maximum in $L_0.$ Hence
\[
\max_{L<-H^*\log_q (A) H-(q-1)^{-1}}G(L)= \tr X^{2-q}\left(\log_q X-H^* \log_q(A)H\right),
\]
which proves $ (iii). $
\end{proof}

Setting $H=I$ we obtain in particular

\begin{corollary}\label{3.2}
The equality
\begin{eqnarray}
&& D_{2-q}(X|A)
=
\tr X^{2-q}\left(\log_q X- \log_q A \right)\nonumber
\\[0.5ex]&=&
\max_{L+\log_q A>-(q-1)^{-1}}\bigl\{\tr X +\tr X^{2-q}L-\tr \exp_q (L+\log_q A )\bigr\}
\end{eqnarray}
holds for $q\in[1,2].$
\end{corollary}

Corollary 3.2 may be considered as a variational representation of the Tsallis relative entropy. For $q=1$ we recover the well-known representation
\begin{eqnarray*}
D(X|A)
&=&
\tr X\left(\log X- \log A \right)\nonumber
\\[0.5ex]&=&
\max_{L}\left\{\tr X +\tr X L-\tr \exp (L+\log A )\right\},
\end{eqnarray*}
where the supremum is taken over self-adjoint $ L. $

\section{Variant representations related to Tsallis relative entropy}

In this section, we generalize the Gibbs variational representations and the variational representations in terms of the quantum relative entropy obtained by Hiai and Petz \cite{HP93}. We recall the Peierls-Bogolyubov type inequalities for deformed exponentials and quote from \cite[Theorem 7]{HLS17}.
\begin{lemma}\label{Peierls-Bogolyubov type inequalities}
Let $A $ and $ B $ be self-adjoint $ n\times n $ matrices. The following assertions hold:
\begin{enumerate}
\item[(i)]  If $q<1,$ and both $A$ and $A+B$ are bounded from above by $-(q-1)^{-1},$ then
\[
\log_q\tr\exp_q(A+B)-\log_q\tr\exp_q A
\ge
\bigl(\tr \exp_q A\bigr)^{q-2}\tr (\exp_q A)^{2-q}B.
\]
\item[(ii)]  If $1<q\le 2,$ and both $A$ and $A+B$ are bounded from below by $-(q-1)^{-1},$ then
\[
\log_q\tr\exp_q(A+B)-\log_q\tr\exp_q A
\ge
\bigl(\tr \exp_q A\bigr)^{q-2}\tr (\exp_q A)^{2-q}B.
\]
\item[(iii)]  If $q\ge 2,$ and both $A$ and $A+B$ are bounded from below by $-(q-1)^{-1},$ then
\[
\log_q\tr\exp_q(A+B)-\log_q\tr\exp_q A
\le
\bigl(\tr \exp_q A\bigr)^{q-2}\tr (\exp_q A)^{2-q}B.
\]
\end{enumerate}
\end{lemma}

Using these Peierls-Bogolyubov type inequalities we obtain:

\begin{theorem}\label{4.2}
The following variational representations hold:
\begin{enumerate}
\item[(i)]  If $q<1,$ then for $L< -(q-1)^{-1},$
\[
\log_q\tr \exp_q L=\max_{X>0, \,\tr X=1}\left\{\tr X^{2-q}L-\tr X^{2-q}\log_q X  \right\},
\]
and for $X>0$ with $\tr X=1,$
\[
\tr X^{2-q}\log_q X=\max_{L< -(q-1)^{-1}}\left\{\tr X^{2-q}L-\log_q\tr \exp_q L \right\}.
\]
\item[(ii)]  If $1<q\le 2,$ then for $L> -(q-1)^{-1},$
\[
\log_q\tr \exp_q L=\max_{ X>0, \,\tr X=1}\left\{\tr X^{2-q}L-\tr X^{2-q}\log_q X \right\},
\]
and for $X>0$ with $\tr X=1,$
\[
\tr X^{2-q}\log_q X=\max_{ L> -(q-1)^{-1}}\left\{\tr X^{2-q}L-\log_q\tr \exp_q L \right\}.
\]
\item[(iii)]  If $q>2,$ then for $L> -(q-1)^{-1},$
\[
\log_q\tr \exp_q L=\min_{X>0, \,\tr X=1}\left\{\tr X^{2-q}L-\tr X^{2-q}\log_q X \right\},
\]
and for $X>0$ with $\tr X=1,$
\[
\tr X^{2-q}\log_q X=\min_{L> -(q-1)^{-1}}\left\{\tr X^{2-q}L-\log_q\tr \exp_q L \right\}.
\]
\end{enumerate}
\end{theorem}

\begin{proof} We just prove the case of $1<q\le 2.$ For $X>0$ with $ \tr X=1$ and setting $A=\log_q X$ we have $\tr\exp_q A=1.$
By $(ii)$ of Lemma 4.1 we thus obtain
\[
\tr X^{2-q}B\le \log_q\tr \exp_q(\log_q X+B),
\]
which holds for $X>0$ with $\tr X=1$ and $B$ with $\log_q X+B>-(q-1)^{-1}.$ Replacing $B$ with $L-\log_q X$ yields
\begin{eqnarray}
\tr X^{2-q}L\le \log_q\tr \exp_q L+\tr X^{2-q}\log_q X,
\end{eqnarray}
which is valid for $X>0$ with $\tr X=1$ and $L>-(q-1)^{-1}.$
It is easy to see that for a fixed $X,$ there is equality in $(4.1)$ for $L=\log_q X.$ We thus obtain
\[
\tr X^{2-q}\log_q X=\max_{ L>-(q-1)^{-1}}\left\{\tr X^{2-q}L-\log_q\tr \exp_q L \right\}.
\]
By an elementary calculation we obtain the equalities
\[
\log_q{\frac{y}{x}}=\log_q y + y^{q-1}\log_q{\frac{1}{x}},
\]
and
\[
\log_q{\frac{1}{x}}=-x^{1-q}\log_q x
\]
for $q\in\mathbf R.$ Therefore,
\begin{eqnarray*}
\log_q\frac{\exp_q L}{\tr \exp_qL}&=&\log_q\exp_qL+(\exp_q L)^{q-1}\log_q\frac{1}{\tr \exp_qL}\\
&=&
L-(\exp_q L)^{q-1}(\tr \exp_q L)^{1-q}\log_q\tr \exp_q L.
\end{eqnarray*}
It follows that
\begin{eqnarray*}
&& \log_q \tr \exp_q L+\tr \left(\frac{\exp_q L}{\tr \exp_q L}\right)^{2-q}\log_q \left(\frac{\exp_q L}{\tr \exp_q L}\right)
\\
&=&
\log_q\tr \exp_q L+\frac{\tr (\exp_q L)^{2-q}L}{(\tr \exp_q L)^{2-q}}-\frac{\tr \exp_q L (\tr \exp_q L)^{1-q}\log_q \tr \exp_q L}{(\tr \exp_q L)^{2-q}}
\\
&=&
\tr\left(\frac{\exp_q L}{\tr \exp_q L}\right)^{2-q} L.
\end{eqnarray*}
For a fixed $L$ we therefore have equality in $(4.1)$ for $X=(\tr \exp_q L)^{-1} \exp_q L.$
Hence,
\[
\log_q\tr \exp_q L=\max_{ X>0, \,\tr X=1}\left\{\tr X^{2-q}L-\tr X^{2-q}\log_q X \right\}
\]
for $L> -(q-1)^{-1}.$
The cases for $q<1$ and $q>2$ are proved by similar reasoning.
\end{proof}

By setting $L=L+H^*\log_q(Y)H$ in Theorem 4.2 we obtain:

\begin{theorem} \label{4.3}
Assume $H^*H=1.$ The following assertions hold:
\begin{enumerate}
\item[(i)]  If $q<1,$ then for $L\le 0$ we have the equality
\begin{eqnarray*}
&&\log_q\tr \exp_q (L+H^*\log_q YH)
\\&=&\max_{X>0, \,\tr X=1}\left\{\tr X^{2-q}L-\tr X^{2-q}\left(\log_q X-H^*\log_q YH\right) \right\},
\end{eqnarray*}
and for $X>0$ with $\tr X=1$ the equality
\begin{eqnarray*}
&&\tr X^{2-q}\left(\log_q X-H^*\log_q Y H\right)
\\&=&\max_{L\le 0}\left\{\tr X^{2-q}L-\log_q\tr \exp_q (L+H^*\log_q YH) \right\}.
\end{eqnarray*}

\item[(ii)]  If $1<q\le 2,$ then for $L\ge 0$ we have the equality
\begin{eqnarray*}
&&\log_q\tr \exp_q (L+H^*\log_q YH)
\\&=&\max_{ X>0, \,\tr X=1}\left\{\tr X^{2-q}L-\tr X^{2-q}\left(\log_q X-H^*\log_q YH\right) \right\},
\end{eqnarray*}
and for $X>0$ with $\tr X=1$ the equality
\begin{eqnarray*}
&&\tr X^{2-q}\left(\log_q X-H^*\log_q YH\right)
\\&=&\max_{L\ge 0}\left\{\tr X^{2-q}L-\log_q\tr \exp_q (L+H^*\log_q YH)\right\}.
\end{eqnarray*}

\item[(iii)]  If $q>2,$ then for $L\ge 0$ we have the equality
\begin{eqnarray*}
&&\log_q\tr \exp_q (L+H^*\log_q YH)
\\&=&\min_{X>0, \,\tr X=1}\left\{\tr X^{2-q}L-\tr X^{2-q}\left(\log_q X-H^*\log_q YH \right)\right\},
\end{eqnarray*}
and for $X>0,$ $\tr X=1$ the equality
\begin{eqnarray*}
&&\tr X^{2-q}\left(\log_q X-H^*\log_q YH\right)
\\&=&\min_{L\ge 0}\left\{\tr X^{2-q}L-\log_q\tr \exp_q (L+H^*\log_q YH)\right\}.
\end{eqnarray*}
\end{enumerate}
\end{theorem}

\begin{remark} \label{4.4}
We note that the cases $1\le q\le 2$ in Theorem 4.2 and $1\le q\le 2$ with $H=I$ in Theorem 4.3 was first obtained by Furuichi in \cite{Fur06}, who gave a different proof.
Note also that when $q\rightarrow 1,$ we recover Gibbs' variational principle for the von Neumann entropy $S(X)=-\tr X \log X $
together with the variational representations related to the quantum relative entropy obtained by Hiai and Petz,  when $L$ is self-adjoint.
Moreover, we can derive  convexity or concavity of the map
\[
Y\rightarrow \log_q\tr \exp_q (L+H^*\log_q YH)
\]
by using the joint convexity or concavity of the Tsallis entropy type functionals
\[(X,Y)\rightarrow \tr X^{2-q}\left(\log_q X-H^*\log_q YH\right),
\]
which in turn recovers the Peierls-Bogolyubov type inequalities for deformed exponentials. Note that the joint convexity or concavity for the Tsallis entropy type functionals can be traced back to Lieb's concavity theorem and Ando's convexity theorem, as demonstrated in Corollary 2.3.
\end{remark}

Now we consider two types of variational expressions with and without the restriction $\tr X=1.$
A special case of Theorem 2.2 states that for positive numbers $s$ and $\lambda,$
\begin{eqnarray}
\exp_q s=\max_{\lambda>0}\left\{\lambda-\lambda^{2-q}(\log_q \lambda-s)\right\},  \quad q \le 2,
\end{eqnarray}
and
\begin{eqnarray}
\exp_q s=\min_{\lambda>0}\left\{\lambda-\lambda^{2-q}(\log_q \lambda-s)\right\}, \quad q>2,
\end{eqnarray}
which may be viewed as Legendre-Fenchel type dualities for deformed exponentials.
The inequalities $(4.2)$ and $(4.3)$ may also easily be obtained from the scalar Young's inequality and its reverse inequality.
We now recover Theorem 2.2 from Theorem 4.3 and the above scalar Legendre-Fenchel dualities.
For $1\le q\le 2$ and
by using Theorem 4.3 $(ii),$ we obtain
\begin{eqnarray*}
&&\max_{X>0}\left\{\tr X +\tr X^{2-q} L-\tr X^{2-q}\left(\log_q X-H^*\log_q YH\right)\right\}
\\
&=&
\max_{\lambda>0}\max_{\bar{X}>0, \,\tr \bar{X}=1}\left\{\tr \lambda \bar{X} +\tr (\lambda \bar{X})^{2-q} L-\tr (\lambda \bar{X})^{2-q}\left(\log_q (\lambda \bar{X})-H^*\log_q YH\right)\right\}
\\
&=&
\max_{\lambda>0}\max_{\bar{X}>0, \tr \,\bar{X}=1}\left\{\lambda^{2-q}\left(\tr \bar{X}^{2-q}L-\tr \bar{X}^{2-q}(\log_q \bar{X}- H^*\log_q YH)\right)
+\lambda - \lambda^{2-q}\log_q \lambda\right\}
\\
&=&
\max_{\lambda>0}\left\{\lambda^{2-q}\log_q\tr \exp_q (L+H^*\log_q YH)+\lambda- \lambda^{2-q}\log_q \lambda\right\}\\
&=&
\tr \exp_q(L+H\log_q YH),
\end{eqnarray*}
where the last equality follows from $(4.2).$
The cases for $q<1$ and $q>2$ are proved by similar reasoning.

\section{Golden-Thompson's inequality for deformed exponentials}

The second author generalized Golden-Thompson's trace inequality to $q$-exponentials with deformation parameter $q\in[1,3].$ We will now address the same question for parameter values $ q\in [0,1). $
The following result is an easy consequence of Corollary 2.3.

\begin{corollary}\label{5.1} Let $H_1, \ldots, H_k$ be matrices with $H_1^*H_1+\cdots+H_k^*H_k=1.$ The function
\begin{eqnarray}
\varphi(A_1,\ldots, A_k)=\tr \exp_q\left(\sum_{i=1}^{k}H_i^*\log_q(A_i)H_i\right),
\end{eqnarray}
defined in $k$-tuples of positive definite matrices, is concave for $q\in[0,1).$
\end{corollary}

The second author \cite[Theorem 3.1]{Han15} proved that $\varphi$ is positively homogeneous of degree one. Since $\varphi$ is concave for $q\in[0,1)$  and by appealing to \cite[Lemma 2.1]{Han15}, we may reason as in \cite[Corollary 3.4]{Han15} to obtain:

\begin{corollary}\label{5.2} The function $\varphi$ defined in $(5.1)$ satisfies the inequality
\begin{eqnarray*}
\varphi(B_1,\ldots,B_k)\le \tr \exp_q \left(\sum_{i=1}^kH^*\log_q(A_i)H_i\right)^{2-q}\sum_{j=1}^kH_j^*\bigl(\fd{}\log_q(A_j)B_j\bigr)H_j
\end{eqnarray*}
for $0\le q<1.$
\end{corollary}

\begin{theorem}\label{5.3} Let $A$ and $B$ be negative definite matrices. The inequality
\begin{eqnarray*}
\tr \exp_q(A+B)\le \tr \exp_q(A)^{2-q}\left(A(q-1)+\exp_qB\right)
\end{eqnarray*}
then holds for $0\le q<1.$
\end{theorem}

\begin{proof} In Corollary 5.2 we set $k=2,$  $A_1=B_1$ and $A_2=1.$ We then obtain the inequality
\begin{equation}\label{main inequality for k=2}
\varphi(B_1,B_2)\le \tr \exp_q(H_1^*\log_q(B_1)H_1)^{2-q}(H_1^*B_1^{q-1}H_1+H_2^*B_2H_2)
\end{equation}
for  $0\le q<1.$
Furthermore, we set $ H_1=\varepsilon^{1/2} $ for $ 0< \varepsilon <1. $ To fixed negative definite matrices $L_1$ and $L_2$ we may choose
$B_1$ and $B_2$ such that
\[
\begin{array}{rl}
L_1&=H_1^*\log_q(B_1)H_1=\varepsilon \log_q(B_1),\\[1.5ex]
L_2&=H_2^*\log_q(B_1)H_2=(1-\varepsilon) \log_q(B_2).
\end{array}
\]
By inserting these operators in inequality (\ref{main inequality for k=2}) we obtain
\begin{eqnarray*}
&&\tr \exp_q(L_1+L_2)
\\[0.5ex]
&\le&
\tr \exp_q(L_1)^{2-q}\bigl(\varepsilon \exp_q(\varepsilon^{-1}L_1)^{q-1}+(1-\varepsilon)\exp_q((1-\varepsilon)^{-1}L_2)\bigr)
\\[0.5ex]
&=&
\tr \exp_q(L_1)^{2-q}\bigl(L_1(q-1)+\varepsilon+(1-\varepsilon)\exp_q((1-\varepsilon)^{-1}L_2)\bigr).
\end{eqnarray*}
Since $
\lim_{\varepsilon\rightarrow 0}(1-\varepsilon)\exp_q((1-\varepsilon)^{-1}L_2)=\exp_q(L_2),
$ we obtain
\begin{eqnarray*}
\tr \exp_q(L_1+L_2)
\le
\tr \exp_q(L_1)^{2-q}\bigl(L_1(q-1)+\exp_q(L_2)\bigr).
\end{eqnarray*}
Finally, by replacing $L_1$ and $L_2$ with $A$ and $B,$ the assertion follows.
\end{proof}

For  $q=1$ we recover the Golden-Thompson inequality
\begin{eqnarray*}
\tr \exp(A+B)\le \tr \exp(A)\exp(B),
\end{eqnarray*}
firstly only for negative definite operators. However, by adding suitable constants to $ A $ and $ B, $ we obtain the trace inequality for arbitrary self-adjoint operators.
\\[1ex]
{\bf Acknowledgements.}
The first author acknowledges support from the Natural Science Foundation of the Jiangsu Higher Education Institutions of China, Grant No: 18KJB110033.
The second author acknowledges support from the Japanese government Grant-in-Aid for scientific research 17K05267.

\end{document}